\documentclass[runningheads]{llncs}
\usepackage[T1]{fontenc}
\usepackage{amsmath,graphicx,hyperref}
\usepackage{tabularray} 
\UseTblrLibrary{booktabs}
\usepackage{xcolor}  
\usepackage{caption}
\usepackage{multirow,multicol}
\usepackage{amsfonts}
\usepackage{subcaption}
\usepackage{amssymb}
\usepackage{graphicx}
\makeatletter
\def\@evenhead{\slshape\leftmark\hfil} 
\def\@oddhead{\hfil\slshape\rightmark} 
\makeatother
\begin{document}
\renewcommand{\thepage}{\phantom{\arabic{page}}}
\title{ChaRVoC: A Challenge-Response Voice Cancelable Authentication System}
%
%
\author{\begin{tabular}{c}
Phuc-Khang Vo-Hoang\thanks{Two authors have the equal contribution}$^{ \dagger}$  \qquad Hoang C. Ta$^{\star \dagger}$  \qquad Nhien-An Le-Khac$^{\ddagger} $ \\  Dinh-Thuc Nguyen$^{\dagger}$  \qquad Hong-Hanh Nguyen-Le$^{\ddagger}$
\end{tabular}}
\authorrunning{Phuc-Khang Vo-Hoang et al.}
%
\institute{$^{\dagger}$ University of Science, Ho Chi Minh City, Vietnam \\
         $^{\ddagger}$ University College Dublin, School of Computer Science, Dublin, Ireland}
\maketitle              
\begin{abstract}
In this work, we present a Challenge-Response Voice Cancelable authentication system, called \textbf{ChaRVoC}, which provides protection against replay attacks, revocability issues, and template compromise. Our approach integrates three security factors: (1) inherent voice biometric characteristics, (2) user-memorized secret keys enabling template revocability, and (3) dynamic system-generated challenges providing liveness detection. Specifically, we introduce a novel HashGray-XOR scheme which combines a cryptographic hash function with an unrecoverable graycode-based transformation to create secured templates that are mathematically proven to be non-invertible. We compare our methods with existing cancelable biometric methods (WTA, IoM, RoE) on VoxCeleb1, TIMIT, and VOiCES datasets to show the recognition performance of our proposed system. We also show that our system achieves both cancelability and unlinkability properties. 

\end{abstract}
\section{Introduction}
\label{sec:intro}

Voice-based authentication systems have gained remarkable popularity in recent years, driven by the proliferation of voice-controlled digital assistants, smart home devices, and mobile applications \cite{yadav2023voice,guaman2018device}. However, current voice authentication systems \cite{chee2018cancellable,nguyen2025privacy,nautsch2018homomorphic,paulini2016multi} suffer from three fundamental vulnerabilities that limit their deployment in practice. First, voice biometrics are permanently linked to users and cannot be revoked if compromised \cite{bhargav2006privacy}. Second, stored voice templates in databases present targets for attackers, as a single breach can compromise users across multiple systems \cite{chee2018cancellable,nguyen2025privacy}. Third, these systems are vulnerable to replay attacks, where recorded voice samples can be used to impersonate legitimate users \cite{pradhan2019combating}. 

While cancelable biometric techniques have been proposed to protect templates stored in databases, existing approaches either require complex cryptographic operations that are computationally expensive \cite{billeb2015biometric,nautsch2018homomorphic,johnson2013vaulted} or they fail to address the replay attack vulnerability \cite{pradhan2019combating}. Text-dependent speaker verification systems can prevent replay attacks through dynamic phrases \cite{heigold2016end,rahman2018attention}, but they cannot prevent voice biometrics stored in databases from server-specific attacks. Current solutions address these challenges in isolation \cite{mittal2024pitch,yasur2023deepfake,nguyen2025privacy,nautsch2018homomorphic}, resulting in systems that remain vulnerable to at least one attack vector. 

In this work, we present a \textbf{Cha}llenge-\textbf{R}esponse \textbf{Vo}ice \textbf{C}ancelable authentication system, named \textbf{ChaRVoC}, which simultaneously addresses replay attacks, revocability issues, and template compromise. Our approach combines three security factors: (1) the user's voice biometric characteristics, (2) a user-specific secret key (i.e., PIN) that each person memorizes, and (3) a system-generated random challenge that provides replay attack resistance. Critically, we introduce our novel HashGray-XOR scheme that combines the secret key with the voice features to create protected templates that are both cancelable and non-invertible. Particularly, this scheme consists of two functions: (i) a cryptographic hash function that maps the random numeric sequence to a fixed-length binary string and (ii) an unrecoverable graycode-based function that transforms a real-valued voice feature vector into a binary representation. The binary feature and the hashed key are combined using the XOR operation to create the final protected template. We also mathematically prove the non-invertibility of our HashGray-XOR and demonstrate that templates generated from our system are computationally unlinkable. 



\section{Related Work}
\label{sec:related-work}
Existing voice-based authentication systems address replay attacks and template compromise in isolation. Text-dependent speaker verification systems \cite{heigold2016end,rahman2018attention} mitigate replay attacks by requiring users to speak dynamically generated phrases. Zhange et al. \cite{zhang2017hearing} detect replay attacks through analysis of articulatory movements during speech production, distinguishing genuine human speech from loudspeaker reproductions based on physical mouth and tongue dynamics. To protect voice biometrics stored in databases, cryptographic techniques are applied, including Paillier cryptosystem \cite{nautsch2018homomorphic}, fuzzy commitment scheme \cite{billeb2015biometric}, and fuzzy vault scheme \cite{johnson2013vaulted}. Alternatively, non-invertible transformation approaches provide template security through mathematical one-way functions. El-Moneim et al. \cite{elmoneim2022cancellable} achieve template revocability by transforming speech into spectrograms and extracting user-key-dependent patches. Nguyen et al. \cite{honghanhnl2025privacy} proposed a Ranking-of-Element hashing to encode elements based on the count of smaller-valued elements. 

Unlike these isolated solutions, we propose a unified framework addressing replay attacks, template compromise, and revocability simultaneously. While Ceaparu et al. \cite{ceaparu2020multifactor} present a similar multi-factor approach, they require pre-registered personal smartphones, whereas our system automatically generates challenge sequences without additional device requirements, improving deployment flexibility.

\section{Proposed Method}
\label{sec:method}


\subsection{System Overview}
\label{subsecc:overview}

\begin{figure*}[t]
    \centering
    \includegraphics[width=0.98\textwidth]{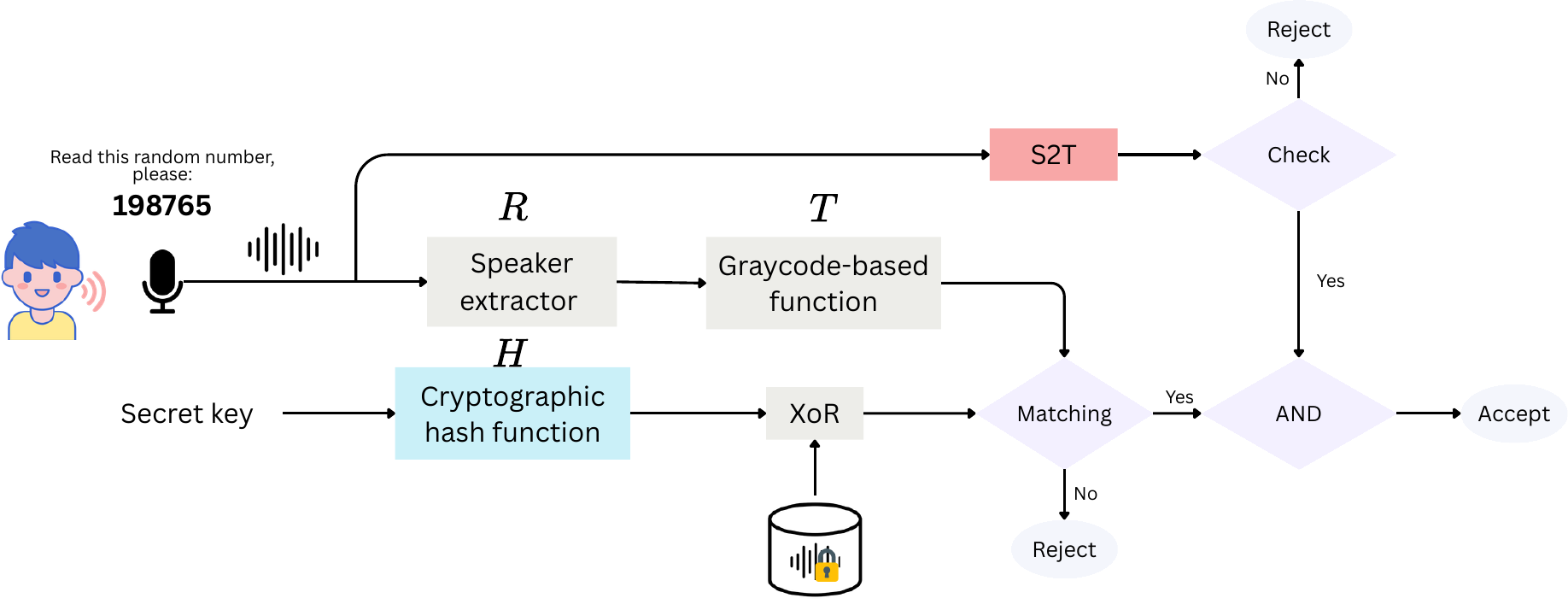}
    \caption{An overview of our Challenge-Response Voice Cancelable Voice Authentication system, namely ChaRVoC.}
    \label{fig:system}
\end{figure*}

As illustrated in Figure \ref{fig:system}, the ChaRVoC system operates through several key components to provide multi-layered security. The Random Challenge Generator produces system-generated numeric sequences (e.g., "198765") that serve as dynamic authentication challenges. These challenges change with each authentication attempt, ensuring that no two authentication sessions use the same input. The speech-to-text (S2T) module verifies that the user correctly speaks the displayed random number, providing real-time liveness detection and preventing replay attacks. Even if an attacker obtains a recording of the user's voice, they cannot authenticate without speaking the current challenge. Simultaneously, a speaker extractor $\mathcal{R}$ transforms raw voice samples into discriminative real-valued feature vectors, capturing the unique voice characteristics of the speaker. The core security mechanism of our system is the proposed HashGray-XoR scheme, which protects biometric templates through two integrated functions: the graycode-based function $\mathcal{T}$ and the cryptographic hash function $\mathcal{H}$.

During enrollment, the user provides a voice sample and a secret key. A speaker extractor $\mathcal{R}$ extracts features from the voice sample, which are transformed into a binary template via the graycode-based function. Simultaneously, the secret key is hashed using $\mathcal{H}$ to produce a binary string. These are combined through XOR to create the protected template. During authentication, the system displays a random challenge that the user reads aloud while providing their secret key. The spoken response is processed through two parallel paths: the S2T module verifies that the correct challenge was spoken, and the speaker recognition model extracts voice features. These features are transformed into a binary template via the graycode-based function $\mathcal{T}$. Meanwhile, the provided secret key is hashed using $\mathcal{H}$, and this hash is XORed with the stored protected template to recover what should be the original enrollment template. This recovered template is then matched against the current voice template. Authentication succeeds only when the S2T verification passes, and the recovered template matches the current voice template. 

\subsection{HashGray-XOR Scheme}
\label{ssec:hashgray}
The HashGray-XOR scheme is formally defined as a transformation $f: \{0, 1\}^* \times I \rightarrow \{0,1\}^n$ that takes as input a secret key $k \in \{0,1\}^*$ and a voice sample $v \in I$ with $I$, where $I$ is the set of all possible voice samples. The output is an $n$-bit binary hash value $t \in \{0,1\}^n$, defined by:
\begin{equation}
    t = f(k,v) = \mathcal{H}(k) \oplus \mathcal{T}(v), \forall k \in \{0,1\}^k, v \in I,
\end{equation}
where $H: \{0,1\}^* \rightarrow \{0,1\}^n$ is the cryptographic hash function, $T: I \rightarrow\{0,1\}^n$ is the unrecoverable graycode-based function, and $\oplus$ denotes the bitwise XOR operation.

\subsubsection{Unrecoverable Graycode-Based Function}
The unrecoverable function $\mathcal{T}$ transforms real-valued voice features into binary representations through a composition function:  $\mathcal{T}(v) = \mathcal{G} \circ r \circ \mathcal{R}(v) = \mathcal{G}(r(\mathcal{R}(v)))$, where $\mathcal{G}: \mathbb{N} \rightarrow \{0,1\}^{l+1}$ is the graycode encoding function, $r: \mathbb{R} \rightarrow \mathbb{Z}$ is the roundoff function defined as $r(x) = \text{round}(x \times 10^p)$ with precision parameter $p$. The sign of each voice feature $v_i = r(f_i)$ after being applied to the roundoff function $r$ is encoded as a binary bit $b_0 = 0$ if $v_i < 0$ and $b_0 = 1$ if $v_i \geq 0$. The integer value $|v_i|$ is converted to an $(l + 1)$-bit Graycode representation.

\begin{theorem} (Unrecoverability Property)
    The function $\mathcal{T}: I \rightarrow {0,1}^n$ constructed as above is unrecoverable, meaning that with $t = \mathcal{T}(v)$, it is computationally infeasible to find back the original $v \in I$.
    \label{theorem:unrecoverability}
\end{theorem}
To prove this theorem, we establish two lemmas:

\begin{lemma} (Non-invertibility of Roundoff Function)
    The function $r: \mathbb{R} \rightarrow \mathbb{Z}$ defined by 
    $$r(x) = \begin{cases} \lfloor x \rfloor, & \text{if } |x - \lfloor x \rfloor| < 0.5 \\
\lceil x \rceil, & \text{otherwise}
\end{cases}$$ is not invertible.
\label{lemma:1}
\end{lemma}

Indeed, if $r$ is invertible, then $r$ is injective, meaning that $\forall x_1, x_2$, if $x_1 \neq x_2$ then $r(x_1) \neq r(x_2)$. This is contradicted by the fact that $r(3.6) = r(3.9) = 4$ but $3.6 \neq 3.9$. 

\begin{lemma}
    Given $f = h \circ g$, the function $f$ is invertible if and only if $h$ and $g$ are invertible functions.
    \label{lemma:2}
\end{lemma}

\begin{proof} \textbf{(Lemma 2)} To prove that $f$ is invertible if and only if $g$ and $h$ are invertible functions, we need to prove two parts. 

\textbf{Part 1}: If $g$ and $h$ are invertible functions. Then there exist inverse functions $g^{-1}$ and $h^{-1}$. The composite function $f = h \circ g$ means $f(x) = h(h(x))$. To prove that $f$ is invertible, we need to find the inverse function $f^{-1}$. 

We have $f^{-1} = h^{-1} \circ g^{-1}$ (by the property of composite functions). Indeed, $(f^{-1} \circ f)(x) = (h^{-1} \circ g^{-1})(x) = x$. Similarly, $(f \circ f^{-1})(x) = x$. Therefore, $f$ is invertible. 

\textbf{Part 2.} If $f$ is invertible, then $g$ and $h$ are invertible. Assume $f$ is invertible, then there exists an invertible function $f^{-1}$. The composite function $f = g \circ h$ means $f(x) = g(h(x)).$ We need to prove that $g$ and $h$ are invertible. 

For $g$, suppose $g$ is not invertible. There exists $x_1 \neq x_2$ such that $g(x_1) = g(x_2).$ Therefore, $f(x_1) = g(h(x_1)) = g(h(x_2)) = f(x_2)$. This contradicts the fact that $f$ is invertible (because $f(x_1) \neq f(x_2)$ when $x_1 \neq x_2$). Therefore, $g$ must be invertible. Similarly, $h$ is proved as the same.

\textbf{Conclusion.} $f$ is invertible if and only if $g$ and $h$ are invertible functions.
\end{proof}

\begin{proof} \textbf{(Theorem 1)}
    
    For any $v \in I$, let $b = (b_0, b_1, ..., b_l) = T(v) = G(r(R(v)))$.  Suppose that given $b = (b_0, b_1, ..., b_l) = T(v) = G(r(R(v)))$, one can recover exactly $v: T(v) = b$. This means given $b = (b_0, b_1, ..., b_l)$, there is only one $v: T(v) = b$, or $T = G \circ H$ is invertible, where $H = r \circ R$. 
    
    By Lemma \ref{lemma:2}, $H = r \circ R$ is invertible if and only if $r$ and $R$ are invertible functions. This contradicts the fact that $r$ is not invertible (by Lemma \ref{lemma:1}). Therefore, $T$ is unrecoverable.
\end{proof}
\subsubsection{Cryptographic Hash Function}
The cryptographic hash function $\mathcal{H}$ is defined as: $\mathcal{H}: \{0,1\}^* \rightarrow \{0,1\}^n$. This function maps an arbitrary binary string (a secret key) to a fixed-size binary string of $n$ bits, serving as the biometric template protection mechanism. 

\begin{definition} (Cancelable Biometric Function) 
A biometric function is cancelable if it can be designed to allow an individual to enroll and revoke many different samples. And the value of the biometric function for a generated sample is called a biometric template.    
\end{definition}

\begin{definition}
    A cryptographic hash function $H: \{0,1\}^* \rightarrow \{0,1\}^n$ has the following desirable properties:
    \begin{enumerate}
        \item Pre-image resistance: Given a hash value $y$, it is computationally difficult to find any $x$ such that $y = H(x)$.
        \item Second pre-image resistance: Given a pre-image $x_1$, it is computationally difficult to find another pre-image $x_2$ such that $H(x_1) = H(x_2)$.
        \item Collision resistance: It is computationally difficult to find two different values $x_1$ and $x_2$ such that $H(x_1) = H(x_2)$.
    \end{enumerate}
\end{definition}

\begin{theorem} (Cancelability Property)
    Given a cryptographic hash function $H: \{0,1\}^* \rightarrow \{0,1\}^n$ and an unrecoverable function of voice biometric $T: I \rightarrow \{0,1\}^n$, the voice biometric function defined as follows is cancelable:
    $f: \{0,1\}^* \times I \rightarrow \{0,1\}^n$ defined by $f(k,v) = H(k) \oplus T(v), \forall k \in \{0,1\}^*, v \in I$. 
    \label{theorem:cancelability}
\end{theorem}
\begin{proof}
    First, we note that if $t_1 = \mathcal{H}(k_1) \oplus \mathcal{T}(v_1)$ and $t_2 = \mathcal{H}(k_2) \oplus \mathcal{T}(v_2)$ are two voice biometric templates of the same person at two different sampling times, $v_1$ and $v_2$, we have:
    
{$
    $$t_1 \oplus t_2 = \begin{cases}
    0, & \text{if } k_1 = k_2, v_1 = v_2 \\
    \mathcal{T}(v_1) \oplus \mathcal{T}(v_2), & \text{if } k_1 = k_2, v_1 \neq v_2 \\
    \mathcal{H}(k_1) \oplus \mathcal{H}(k_2), & \text{if } k_1 \neq k_2, v_1 = v_2 \\
    (\mathcal{H}(k_1) \oplus \mathcal{H}(k_2)) \oplus (\mathcal{T}(v_1) \oplus \mathcal{T}(v_2)), & \text{otherwise}
    \end{cases}$$
$}
\end{proof}
Theorem \ref{theorem:cancelability} shows that even if the voice biometric samples have the same binary representation but different secret keys, the voice biometric function values will be different with high probability if the cryptographic hash function has high collision resistance.

Theorems \ref{theorem:unrecoverability} and \ref{theorem:cancelability} ensure that our ChaRVoC system provides both non-invertible and unrecoverable template protection. 

\subsection{Template Matching}
\label{sec:matching}
The matching process compares binary templates generated during enrollment and authentication. For two binary templates $T_1 = b_{10}b_{11}...b_{1n}$ and $T_2 = b_{20}b_{21}...b_{2n}$, the similarity measure is:
\begin{equation}
    S(T_1, T_2) = \frac{\sum_{i=0}^{n}(1-|b_{1i}-b_{2i}|)}{2n-\sum_{i=0}^{n}(1-|b_{1i}-b_{2i}|)}
\end{equation}
This computes the ratio of matching bits, producing a similarity score between 0 and 1, where 1 indicates identical templates and 0 indicates completely different templates.

\section{Experiments}
\subsection{Experimental Setup}
\textbf{Datasets.} We use three datasets: VoxCeleb1~\cite{voxceleb1}, TIMIT~\cite{timit}, and VOiCES~\cite{richey2018voices} to evaluate our proposed ChaRVoC system. VOiCES is utilized to evaluate the accuracy performance of our system in noisy conditions (e.g., background noise and reverberation).

\textbf{Speaker extractors.} To demonstrate the compatibility of our HashGray-XOR scheme, we utilize three popular speaker recognition models, including ECAPA-TDNN \cite{ECAPA}, RawNet2 \cite{jung2020rawnet2}, and RawNet3 \cite{jung2022raw3}. For all models, we employ publicly available pretrained weights.

\textbf{Baselines.} We compare our proposed method with three cancellable biometrics approaches, Winner-Take-All (WTA) \cite{yagnik2011power}, Index-of-Max (IoM)  \cite{jin2017ranking} hashing, and Ranking-of-Element (RoE) hashing \cite{nguyen2025privacy}.

\textbf{Metrics.} We use three evaluation metrics: equal error rate (EER), the area
under the ROC curve (AUC), and True Match Rate at a False Match Rate of 0.1\% (TMR@FMR=0.1\%). 

\subsection{Experimental Results}
\label{sec:experiment}
This section presents the results of our ChaRVoC system regarding recognition performance and computation efficiency. Furthermore, we also analyze the cancelability and unlinkability of our system.

\begin{table}[!t]
\caption{
Performance Comparison of Different Speaker Recognition Models Across Three Datasets. Lower EER and Higher AUC/TMR Values Indicate Better Performance.}
\label{tab:speakerverification}
\centering
{%
\begin{tabular}{ccccc}
\toprule
\multirow{2}{*}{\textbf{Dataset}} & 
\multirow{2}{*}{\textbf{Model}} & 
\multicolumn{3}{c}{\textbf{Metric}} \\
\cmidrule(lr){3-5}
 & & \textbf{EER} & \textbf{AUC} & \textbf{TMR@FMR=0.1\%} \\
\midrule
\multirow{3}{*}{VoxCeleb1} 
  & RawNet3    & 5.0901      & 0.9836     & \textbf{0.7588} \\
  & RawNet2    & \textbf{2.8791} & \textbf{0.9953} & 0.7509 \\
  & ECAPA-TDNN & 11.0869     & 0.9564     & 0.4045 \\
\midrule
\multirow{3}{*}{TIMIT} 
  & RawNet3    & 2.4901      & 0.9854     & \textbf{0.9272} \\
  & RawNet2    & \textbf{1.7008} & \textbf{0.9986} & 0.8786 \\
  & ECAPA-TDNN & 7.0233      & 0.9819     & 0.5774 \\
\midrule
\multirow{3}{*}{VOiCES} 
  & RawNet3    & 26.1745     & 0.7796     & 0.0250 \\
& RawNet2    & \textbf{11.3680} & \textbf{0.9537} & 0.4774 \\
  & ECAPA-TDNN & 11.6902     & 0.9484     & \textbf{0.6016} \\
\bottomrule
\end{tabular}}
\end{table}

\textbf{Recognition Performance.} Table \ref{tab:speakerverification} shows the performance of our ChaRVoC system with different speaker extractors $\mathcal{R}$. RawNet2 consistently achieves the best overall performance across datasets, with the lowest EER of $2.88\%$ on VoxCeleb1, $1.70\%$ on TIMIT, and $11.368\%$ on the noisy VOiCES dataset. Based on these results, we select RawNet2 as our primary speaker extractor for subsequent experiments.

\textbf{Comparison with Baselines.} Table \ref{tab:baseline} demonstrates that our method significantly outperforms existing cancelable biometric approaches across all datasets. Unlike WTA and IoM, which may lose essential discriminative information
during the hashing process, our approach preserves the biometric characteristics. Compared to the state-of-the-art voice-based cancelable biometrics system - RoE, our method also outperforms with the EER of $2.879\%$, $1.700\%$, and $11/368\%$ on VoxCeleb1, TIMIT, and VOiCES datasets, respectively. 

\begin{table}[t]
\centering
\caption{Performance comparison of our ChaRVoC method against baselines (WTA, IoM, RoE). Lower EER and higher AUC/TMR values indicate better performance.}
\label{tab:baseline}
{%
\begin{tabular}{ccccc}
\toprule
\multirow{2}{*}{\textbf{Dataset}} & 
\multirow{2}{*}{\textbf{Method}} & 
\multicolumn{3}{c}{\textbf{Metric}} \\
\cmidrule(lr){3-5}
 & & \textbf{EER} & \textbf{AUC} & \textbf{TMR@FMR=0.1\%} \\
\midrule
\multirow{4}{*}{VoxCeleb1} 
  & WTA          & 9.341          & 0.798          & 0.576 \\
  & IoM          & 8.260          & 0.856          & 0.631 \\
  & RoE          & 3.213          & 0.959          & 0.713 \\
  & \textbf{Our} & \textbf{2.879} & \textbf{0.995} & \textbf{0.751} \\
\midrule
\multirow{4}{*}{TIMIT} 
  & WTA          & 8.624          & 0.806          & 0.8205 \\
  & IoM          & 7.256          & 0.891          & 0.8384 \\
  & RoE          & 2.294          & 0.976          & 0.8491 \\
  & \textbf{Our} & \textbf{1.700} & \textbf{0.998} & \textbf{0.878} \\
\midrule
\multirow{4}{*}{VOiCES} 
  & WTA          & 18.041         & 0.578          & 0.301 \\
  & IoM          & 15.297         & 0.687          & 0.398 \\
  & RoE          & 11.879         & 0.894          & 0.461 \\
  & \textbf{Our} & \textbf{11.368}& \textbf{0.953} & \textbf{0.477} \\
\bottomrule
\end{tabular}}
\end{table}


\begin{table}[h!]
    \centering
    \caption{Comparison of Time Complexity (in seconds) for Template Generation between ChaRVoC and Baselines.}
    \label{tab:time_complexity}
    \begin{tabular}{llc}
        \toprule
        \textbf{Method} & \textbf{Dataset} & \textbf{Time complexity (s)} \\
        \midrule
        \multirow{2}{*}{WTA [24]} 
         & VoxCeleb1 & 0.0380 \\
         & TIMIT     & 0.0380 \\
        \midrule
        \multirow{2}{*}{IoM [25]} 
         & VoxCeleb1 & 0.1050 \\
         & TIMIT     & 0.1070 \\
        \midrule
        \multirow{2}{*}{RoE [4]} 
         & VoxCeleb1 & 0.1060 \\
         & TIMIT     & 0.1010 \\
        \midrule
        \multirow{2}{*}{\textbf{ChaRVoC (Ours)}} 
         & \textbf{VoxCeleb1} & \textbf{0.0054} \\
         & \textbf{TIMIT}     & \textbf{0.0052} \\
        \bottomrule
    \end{tabular}
\end{table}


\textbf{Time Complexity Evaluation.} Table \ref{tab:time_complexity} presents a comparison of the time complexity for template generation between our ChaRVoC system and the baseline methods (WTA, IoM, and RoE) on the VoxCeleb1 and TIMIT datasets. Our method demonstrates superior computational efficiency, with template generation times of just 0.0054 seconds on VoxCeleb1 and 0.0052 seconds on TIMIT, significantly faster than all the baselines.


\begin{figure*}[ht]
        \centering
        \caption{Distributions of mated samples and non-mated samples.}
        \label{fig:unlinkability}
        \begin{subfigure}{0.32\textwidth}
            \centering
            \includegraphics[width=\textwidth]{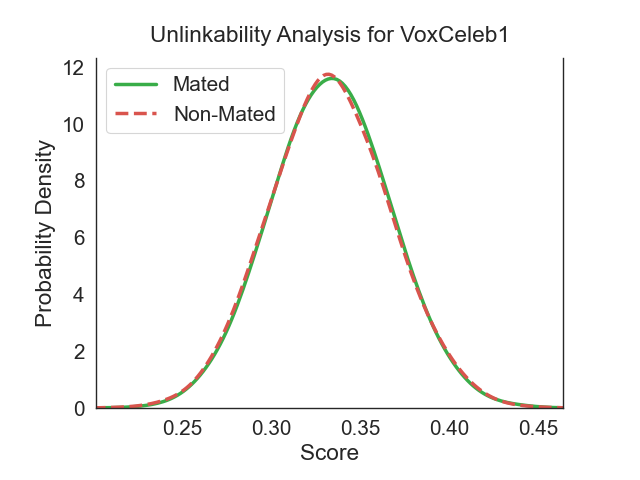}
            \caption{VoxCeleb1}
        \end{subfigure}\hfill
        \begin{subfigure}{0.32\textwidth}
            \centering
            \includegraphics[width=\textwidth]{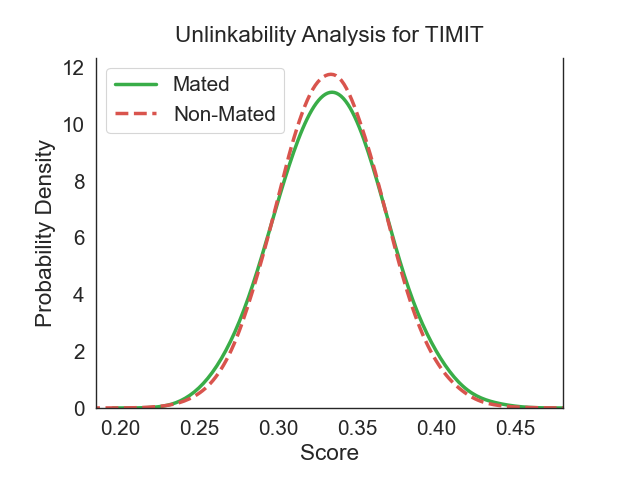}
            \caption{TIMIT}
        \end{subfigure}\hfill
        \begin{subfigure}{0.32\textwidth}
            \centering
            \includegraphics[width=\textwidth]{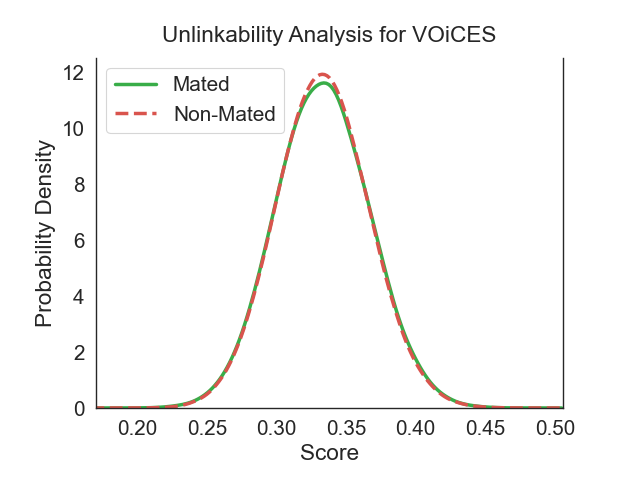}
            \caption{VOiCES}
        \end{subfigure}
\end{figure*}

\textbf{Security Analysis.} To evaluate unlinkability, we used the framework proposed by Gomez-Barrero et al. \cite{gomez2017general}. Figure \ref{fig:unlinkability} illustrates the similarity score distributions between mated (same voice, different keys) and non-mated (different voices) samples. The substantial overlap between these distributions across all datasets confirms that templates generated from the same voice but different keys are computationally indistinguishable from those of different users, ensuring strong privacy protection against cross-matching attacks. 

\section{Conclusion}
\label{sec:conclusion}

In this work, we propose \textbf{ChaRVoC}, a Challenge-Response Voice Cancelable authentication system. By integrating voice biometrics, secret keys, and dynamic challenge-response mechanisms, our system provides template revocability, replay attack prevention, and non-invertible protection through the novel HashGray-XOR scheme. Experimental results across VoxCeleb1, TIMIT, and VOiCES datasets demonstrate superior performance of our system over existing cancelable biometric systems, achieving EERs as low as $1.70\%$. Furthermore, our security analysis confirms strong template unlinkability, while practical processing times (0.76s enrollment, 3.18s authentication) validate real-world deployability.
\bibliographystyle{splncs04}
\bibliography{refs}
\newpage
\end{document}